\documentclass[sigconf]{acmart}

\AtBeginDocument{%
	}

\copyrightyear{2023}
\acmYear{2023}
\setcopyright{rightsretained}
\acmConference[MM '23]{Proceedings of the 31st ACM International Conference on Multimedia}{October 29-November 3, 2023}{Ottawa, ON, Canada}
\acmBooktitle{Proceedings of the 31st ACM International Conference on Multimedia (MM '23), October 29-November 3, 2023, Ottawa, ON, Canada}
\acmDOI{10.1145/3581783.3611943}
\acmISBN{979-8-4007-0108-5/23/10}

\makeatletter
\gdef\@copyrightpermission{
	\begin{minipage}{0.3\columnwidth}
		\href{https://creativecommons.org/licenses/by/4.0/}{\includegraphics[width=0.90\textwidth]{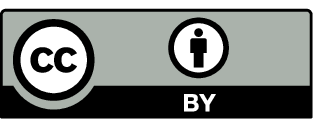}}
	\end{minipage}\hfill
	\begin{minipage}{0.7\columnwidth}
		\href{https://creativecommons.org/licenses/by/4.0/}{This work is licensed under a Creative Commons Attribution International 4.0 License.}
	\end{minipage}
	\vspace{5pt}
}
\makeatother

\usepackage{booktabs}           
\usepackage{multirow}           
\usepackage{amsmath, amsfonts}           
\usepackage{graphicx}           
\usepackage{duckuments}         
\usepackage{subcaption}
\usepackage{bm}
\usepackage{xspace}
\usepackage{xurl}

\newcommand{\ie}{\emph{i.e.},\xspace}
\newcommand{\eg}{\emph{e.g.},\xspace}

\newcommand{\fm}{\texttt{FREEDOM}}

\newtheorem{theorem}{Theorem}[section]

\newtheorem{lemma}[theorem]{Lemma}

\begin{document}

\title{A Tale of Two Graphs: Freezing and Denoising Graph Structures for Multimodal Recommendation}

\author{Xin Zhou}
\affiliation{%
	\institution{Alibaba-NTU Singapore Joint Research Institute \\Nanyang Technological University}
	\country{Singapore}
}
\email{xin.zhou@ntu.edu.sg}

\author{Zhiqi Shen}
\affiliation{%
	\institution{School of Computer Science and Engineering \\Nanyang Technological University}
	\country{Singapore}}
\email{zqshen@ntu.edu.sg}

\begin{abstract}
	Multimodal recommender systems utilizing multimodal features (\eg images and textual descriptions) typically show better recommendation accuracy than general recommendation models based solely on user-item interactions. 
	Generally, prior work fuses multimodal features into item ID embeddings to enrich item representations, thus failing to capture the latent semantic item-item structures.
	In this context, LATTICE proposes to learn the latent structure between items explicitly and achieves state-of-the-art performance for multimodal recommendations.
	However, we argue the latent graph structure learning of LATTICE is both inefficient and unnecessary.
	Experimentally, we demonstrate that freezing its item-item structure before training can also achieve competitive performance.
	Based on this finding, we propose a simple yet effective model, dubbed as \fm{}, that \underline{FREE}zes the item-item graph and \underline{D}en\underline{O}ises the user-item interaction graph simultaneously for \underline{M}ultimodal recommendation. 
	Theoretically, we examine the design of \fm{} through a graph spectral perspective and demonstrate that it possesses a tighter upper bound on the graph spectrum.
	In denoising the user-item interaction graph, we devise a degree-sensitive edge pruning method, which rejects possibly noisy edges with a high probability when sampling the graph.
	We evaluate the proposed model on three real-world datasets and show that \fm{} can significantly outperform current strongest baselines. 
	Compared with LATTICE, \fm{} achieves an average improvement of 19.07\% in recommendation accuracy while reducing its memory cost up to 6$\times$ on large graphs. The source code is available at: \url{https://github.com/enoche/FREEDOM}.
\end{abstract}



\maketitle

\section{Introduction}
The increasing availability of multimodal information (\eg images, texts and videos) associated with items enables users to have a more in-depth and comprehensive understanding of the items they are interested in.
Consequently, multimodal recommender systems (MRSs) leveraging multimodal information are a recent trend in capturing the preferences of users for accurate item recommendations in online platforms, such as e-commerce and instant video platforms, etc.
Studies~\cite{he2016vbpr,zhang2021mining} demonstrate that MRSs usually have better performance than general recommendation models which only utilize the historical user-item interactions.

The primary challenge in MRSs is how to effectively integrate the multimodal information into the collaborative filtering (CF) framework.
Conventional studies either fuse the projected low-dimensional multimodal features with the ID embeddings of items via concatenation, summation operations~\cite{he2016vbpr, liu2017deepstyle} or leverage attention mechanisms to capture users' preferences on items~\cite{chen2017attentive, liu2019user, chen2019personalized}.
The surge of research on graph-based recommendations~\cite{wu2020graph, zhang2022diffusion, zhou2023selfcf, zhou2023layer} inspires a line of work~\cite{wei2020graph, wang2021dualgnn, zhang2021mining, zhou2023enhancing} that exploits the power of graph neural networks (GNNs) to capture the high-order semantics between multimodal features and user-item interactions.
Specifically, MMGCN~\cite{wei2019mmgcn} utilizes graph convolutional networks (GCNs) to propagate and aggregate information on every modality of the item.
On top of MMGCN, GRCN~\cite{wei2020graph} refines the user-item bipartite graph by weighing the user-item edges according to the affinity between user preference and item content.
To achieve better recommendation performance, researchers exploit auxiliary graph structures to explicitly capture the relations between users or items.
For instance, DualGNN~\cite{wang2021dualgnn} constructs the user-user relation graph to smoothen users' preferences with their neighbors via GNNs.
Nonetheless, the authors of LATTICE~\cite{zhang2021mining} argue that the latent item-item structures underlying the multimodal contents of items could lead to better representation learning.
Hence, LATTICE first dynamically constructs a latent item-item graph by considering raw and projected multimodal features learned from multi-layer perceptrons (MLPs). 
It then performs graph convolutions on the constructed latent item-item graph to explicitly incorporate item relationships into representation learning.
As a result, LATTICE could exploit both the high-order interaction semantic from the user-item graph and the latent item content semantic from the item-item structure.
Although this paradigm turns out to be effective, it poses a prohibitive cost by requiring computation and memory quadratic in the number of items. 

In this paper, we first experimentally disclose that the item-item structure learning of LATTICE is dispensable.
Specifically, we build an item-item graph directly from the raw multimodal contents of items and freeze it in training LATTICE.
Under the same evaluation settings and datasets of LATTICE, our empirical experiments in Table~\ref{tab:lattice-frozen} show that LATTICE with frozen item-item graph structures (\ie LATTICE-Frozen) gains slightly better performance on two out of three datasets.
To further refine the structures in both user-item bipartite graph and the item-item graph, we propose a graph structures \underline{FREE}zing and \underline{D}en\underline{O}ising \underline{M}ultimodal model for recommendation, dubbed as \fm{}.
To be specific, the item-item graph constructed by LATTICE retains the affinities between items as the edge weights, which may be noisy as the multimodal features are extracted from general pre-trained models (\eg Convolutional Neural Networks or Transformers).
\fm{} discretizes the weighted item-item graph into an unweighted one to enable information propagation in GCNs only depending on graph structure.
For the user-item graph, we further introduce a degree-sensitive edge pruning technique to denoise its structure to remove the noise caused by unintentional interactions or bribes~\cite{zhou2022bribery}.
Inspired by~\cite{chen2020simple}, we sample edges following a multinomial distribution with pre-calculated parameters to construct a sparsified subgraph. 
Finally, \fm{} learns the representations of users and items by integrating the unweighted item-item graph and the sparsified user-item subgraph.
We analyze \fm{} via a spectral perspective to demonstrate that \fm{} is capable of achieving a tighter low-pass filter.
We conduct comprehensive experiments on three real-world datasets to show that our proposed model can significantly outperform the state-of-the-art methods in terms of recommendation accuracy.

\begin{table}[bpt]
	\centering		
	\def\arraystretch{0.9}	
	\caption{LATTICE with frozen item-item graph structures (\ie LATTICE-Frozen) shows slightly better performance than its original version on Baby and Sports in terms of Recall and NDCG (Refer Section~\ref{sec:experiments} for detailed experiment settings).}
	\begin{tabular}{l l c c}
		\toprule
		Dataset & Metric & LATTICE & LATTICE-Frozen \\
		\midrule
		\multirow{4}{*}{Baby} & R@10 & 0.0547 & 0.0551 \\
		& R@20 & 0.0850 & 0.0873 \\
		& N@10 & 0.0292 & 0.0291 \\
		& N@20 & 0.0370 & 0.0373 \\
		\midrule
		\multirow{4}{*}{Sports} & R@10 & 0.0620 & 0.0626 \\
		& R@20 & 0.0953 & 0.0964 \\
		& N@10 & 0.0335 & 0.0336 \\
		& N@20 & 0.0421 & 0.0423 \\
		\midrule
		\multirow{4}{*}{Clothing} & R@10 & 0.0492 & 0.0434\\
		& R@20 & 0.0733 & 0.0635 \\
		& N@10 & 0.0268 & 0.0227 \\
		& N@20 & 0.0330 & 0.0279 \\
		\bottomrule
	\end{tabular}
	\label{tab:lattice-frozen}
	\vspace{-15pt}
\end{table}

\section{Related Work}
\label{sec:relatedwork}

\subsection{Multimodal Recommendation}
Multimodal recommendation models utilize deep or graph learning techniques to effectively incorporate the multimodal information of items into the classic CF paradigm for better recommendation performance.
Previous work~\cite{he2016vbpr, liu2017deepstyle} extends the BPR method~\cite{rendle2009bpr} by fusing the visual and/or style content of items with their ID embeddings and obtains a considerable performance boost.
To further differentiate users' preferences on multimodal information, attention mechanisms are exploited in multimodal recommendation models.
For instance, VECF~\cite{chen2019personalized} utilizes the VGG model~\cite{simonyan2014very} to perform pre-segmentation on images and captures the user's attention on different image regions.
MAML~\cite{liu2019user} uses a two-layer neural network to capture the user's preference for textual and visual features of an item.
Under the surge of GNNs applied in recommendation systems~\cite{wu2020graph}, researchers are inspired to inject high-order semantics into user/item representation learning via GNNs.
MMGCN~\cite{wei2019mmgcn} adopts the message passing mechanism of GCNs and constructs a modality-specific user-item bipartite graph, which can capture the information from multi-hop neighbors to enhance the user and item representations.
Following MMGCN, GRCN~\cite{wei2020graph} introduces a graph refine layer to refine the structure of the user-item interaction graph by identifying the noise edges and corrupting the false-positive edges. 
In DualGNN~\cite{wang2021dualgnn}, the authors argue that users' preferences may dynamically evolve with time. They introduce a user co-occurrence graph with a preference learning module to capture the user's preference for features from different modalities of an item. 
However, the aforementioned models mine the semantic information between items in an implicit manner and may lead to inferior performance.
To this end, LATTICE~\cite{zhang2021mining} explicitly constructs item-item relation graphs for each modality and fuses them together to obtain a latent item-item graph. 
It then dynamically updates the item-item graph with projected multimodal features from MLPs and achieves state-of-the-art recommendation accuracy.
Recently, we also see an emerging of applying self-supervised learning in multimodal recommendations~\cite{tao2022self, zhou2023bootstrap}.
For an in-depth exploration of multimodal recommender systems, we recommend consulting the comprehensive survey conducted by ~\cite{zhou2023comprehensive}.

\subsection{Denoising Graph Structures}
Studies on denoising graph structures can be roughly categorized as either permanent or temporary edge pruning.
Edges are either pruned or retained based on calculated scores.
The former work~\cite{wang2021denoising, wang2021dualgnn, luo2021learning, xu2022graph} leverages attention mechanisms or pre-defined functions to calculate a node's attention or affinity scores with its neighbors.
Based on the scores, they either permanently prune the edges satisfying a pre-defined condition or decrease the weights on the edges.
The latter work~\cite{hamilton2017inductive, rong2020dropedge} iteratively prunes the edges of the graph to extract a smaller subgraph for graph learning~\cite{luo2021learning}. 
It is worth mentioning that DropEdge~\cite{rong2020dropedge} is an intuitive and widely-used method that prunes edges following a uniform distribution. 
Other studies~\cite{franceschi2019learning, kazi2022differentiable} place a parameterized distribution or function on edges of the graph and learn the parameters in a supervised manner.
A sparsified graph is sampled from the original graph following the learned distributions or functions.
For computational complexity consideration, in this paper, we precalculate and fix the parameter of distribution for subgraph sampling.

\begin{figure*}[t]
	\centering
	\includegraphics[width=0.96\textwidth]{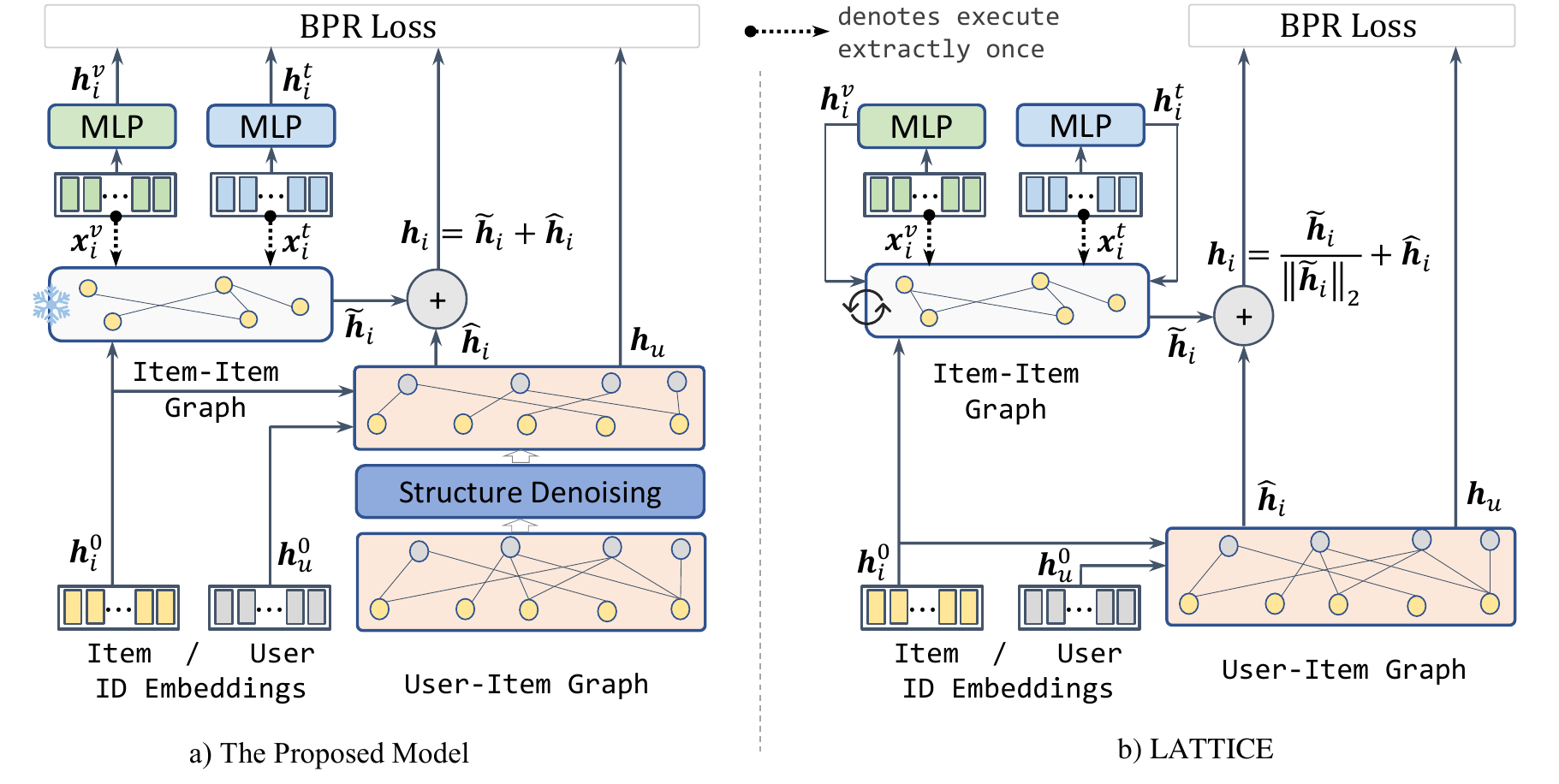} 
	\caption{Comparison of our proposed a) \fm{} and b) LATTICE~\cite{zhang2021mining}. \fm{} freezes the item-item graph and denoises the user-item graph simultaneously for multimodal recommendation.}
	\label{fig:framework}
\end{figure*}

\section{Freezing and Denoising Graph Structures}
\label{sec:freedom}
In this section, we elaborate on each component of \fm{} from graph construction to item recommendation. Fig.~\ref{fig:framework} shows the overall architecture compared with LATTICE.
\subsection{Constructing Frozen Item-Item Graph}
Following~\cite{zhang2021mining}, \fm{} also uses $k$NN to construct an initial modality-aware item-item graph $\bm{S}^m$ using raw features from each modality $m$.
Considering $N$ items, we calculate the similarity score $\bm{S}^m_{ij}$ between item pair $i$ and $j$ with a cosine similarity function on their raw features ($\bm{x}_i^m$ and $\bm{x}_j^m$) of modality $m$. That is:
\begin{equation}
	\bm{S}_{i j}^m = \frac{{(\bm{x}_i^m)}^\top \bm{x}_j^m}{ \| \bm{x}_i^m \| \| \bm{x}_j^m \|},
	\label{eq:sim}
\end{equation}
where $\bm{S}^m_{ij}$ is the $i$-th row, $j$-th column element of matrix $\bm{S}^m \in \mathbb{R}^{N \times N}$.
We further employ $k$NN sparsification~\cite{chen2009fast} and convert the weighted $\bm{S}^m$ into an unweighted matrix. That is, for each item $i$, we only retain the connection relations of its top-$k$ similar edges:
\begin{equation}
	\widehat{\bm{S}}^m_{ij}=\begin{cases}
		1, \enspace & {\bm{S}}^m_{i j} \in \operatorname{top-}k({\bm{S}}^m_{i}), \\
		0, \enspace & \text{otherwise}.
	\end{cases}
	\label{eq:topk}
\end{equation}
Each element in $\widehat{\bm{S}}^m$ is either 0 or 1, with 1 denoting a latent connection between the two items.  We empirically fixed the value of $k$ at 10.
Note that $\widehat{\bm{S}}^m$ is different from the weighted similarity matrix of LATTICE, which uses the affinity values between items as its elements.
We normalize the discretized adjacency matrix $\widehat{\bm{S}}^m$ as $\widetilde{\bm{S}}^m = ({\bm{D}^m})^{-\frac{1}{2}} \widehat{\bm{S}}^m ({\bm{D}^m})^{-\frac{1}{2}}$, where $\bm{D}^m \in \mathbb{R}^{N \times N}$ is the diagonal degree matrix of $\widehat{\bm{S}}^m$ and $\bm{D}_{ii}^m = \sum_{j}\widehat{\bm{S}}^m_{i j}$. 
With the resulted modality-aware adjacency matrices, we construct the latent item-item graph by aggregating the structures from each modality:
\begin{equation}
	{\bm{S}} = \sum_{m \in \mathcal{M}} \alpha_{m} \widetilde{\bm{S}}^m,
\end{equation}
where $\bm{S} \in \mathbb{R}^{N \times N}$, \(\alpha_m\) is the importance score of modality $m$ and $\mathcal{M}$ is the set of modalities. 
Same as other studies~\cite{liu2019user,zhang2021mining}, we consider visual and textual modalities denoted by $\mathcal{M} = \{v, t\}$ in this paper.
The importance score can be learned via parametric functions. 
Here, we reduce the model parameters by introducing a hyperparameter \(\alpha_v\) denoting the importance of visual modality in constructing \(\bm{S}\). We let \(\alpha_t = 1 - \alpha_v\).

Finally, we freeze the latent item-item graph, which can greatly improve the efficiency of \fm{}.
To be specific, the construction of the similarity matrix under modality $m$ in Eq.~\eqref{eq:sim} requires computational complexity of $\mathcal{O}(N^2 d_m)$, where $d_m$ is the dimension of raw features.
As shown in Fig.~\ref{fig:framework}, \fm{} pre-calculates (dotted lines) the item-item graph before training and freezes it during model training. 
As a result, it removes the computational burden of $\mathcal{O}(N^2 d_m)$ for graph construction in training.

\subsection{Denoising User-Item Bipartite Graph}
In this section, we introduce a degree-sensitive edge pruning to denoise the user-item bipartite graph.
The idea is derived from recent researches on model sparsification~\cite{rong2020dropedge} and simplification~\cite{chen2020simple}.
Specifically, DropEdge~\cite{rong2020dropedge} randomly drops a certain ratio of edges in training. In~\cite{chen2020simple}, the authors verify that popular nodes are more likely to suffer from over-smoothing.
Inspired by the finding, we sparsify the graph by pruning superfluous edges following a degree-sensitive probability. 

Formally, we denote a user-item graph as $\mathcal{G} = (\mathcal{V},\mathcal{E})$, where $\mathcal{V}$ is the set of nodes and $\mathcal{E}$ is the set of edges. 
The number of users and items in the user-item graph is $M$ and $N$, respectively. We have $M+N = |\mathcal{V}|$, where $|\cdot|$ denotes the cardinality of a set.
We construct a symmetric adjacency matrix $\bm{A} \in \mathbb{R}^{|\mathcal{V}| \times |\mathcal{V}|}$ from the user-item interaction matrix 
$\bm{R} \in \mathbb{R}^{M \times N}$:
\begin{equation}
	\bm{A} =
	\begin{pmatrix}
		\bm{0} & \bm{R} \\
		\bm{R}^\top & \bm{0}
	\end{pmatrix},
	\label{eq:adj-matrix}
\end{equation}
and each entry $\bm{A}_{ui}$ of $\bm{A}$ is set to 1, if user $u$ has interacted with item $i$, otherwise, $\bm{A}_{ui}$ is set to 0. 

Given a specific edge $e_{k} \in \mathcal{E}, (0 \le k < |\mathcal{E}|)$ which connects node $i$ and $j$, we calculate its probability as $p_k = \frac{1}{\sqrt {\omega_i} \sqrt {\omega_j}}$,
where $\omega_i$ and $\omega_j$ are the degrees of nodes $i$ and $j$ in graph $\mathcal{G}$, respectively.
Usually, we prune a certain proportion $\rho$ of edges of the graph.
That is, the number of edges should be pruned is $\lfloor \rho |\mathcal{E}| \rfloor$, where $\lfloor \cdot \rfloor$ is the floor function.
As a result, the number of retained edges is $n = \lceil |\mathcal{E}| (1 - \rho) \rceil$.
Thus, we sample edges from the multinomial distribution with index $n$ and parameter vector $\bm{p}= \langle p_0, p_1, \cdots, p_{|\mathcal{E}|-1} \rangle$. 
In this way, edges connecting high-degree nodes have a low probability to be sampled from the graph. That is, these edges are more likely to be pruned in $\mathcal{G}$.
We then construct a symmetric adjacency matrix $\bm{A}_{\rho}$ based on the sampled edges following Eq.~\eqref{eq:adj-matrix}.
In line with prior latent item-item graph, we also perform the re-normalization trick on $\bm{A}_{\rho}$, resulting as $\widehat{\bm{A}}_{\rho}$.
Same as DropEdge, \fm{} prunes the user-item graph and normalizes the sampled adjacency matrix iteratively in each training epoch.
However, we resort to the original normalized adjacency matrix $\widehat{\bm{A}} = \bm{D}^{-1/2} \bm{A} \bm{D}^{-1/2}$ in model inference.

\subsection{Integration of Two Graphs for Learning}
We perform graph convolutions on both graphs, that is, we employ a light-weighted GCN~\cite{he2020lightgcn} for information propagation and aggregation on $\bm{S}$ and $\widehat{\bm{A}}_{\rho}$. Specifically, the graph convolution over the item-item graph is defined as:
\begin{equation}
	\widetilde{\bm{h}}_{i}^{l} = \sum_{j \in \mathcal{N}(i)} {\bm{S}}_{ij} \widetilde{\bm{h}}_{j}^{l-1},
\end{equation}
where $\mathcal{N}(i)$ is the neighbor items of $i$, $\widetilde{\bm{h}}_{i}^{l} \in \mathbb{R}^{d}$ is the $l$-th layer item representation of item $i$, $\widetilde{\bm{h}}_{i}^{0}$ denotes its corresponding ID embedding vector and $d$ is the dimension of an item or user ID embedding. 
We stack $L_{ii}$ convolutional layers on the item-item graph $\bm{S}$ and obtain the last layer representation $\widetilde{\bm{h}}_{i}^{L_{ii}}$ as the representation $\widetilde{\bm{h}}_i \in \mathbb{R}^{d}$ of $i$ from the multimodal view:
\begin{equation}
	\widetilde{\bm{h}}_i = \widetilde{\bm{h}}_{i}^{L_{ii}}.
	\label{eq:ii_emb}
\end{equation}
Analogously, in the user-item graph, we perform $L_{ui}$ convolutional operations on $\widehat{\bm{A}}_{\rho}$ and obtain embedding of a user $\widehat{\bm{h}}_u \in \mathbb{R}^{d}$ or an item $\widehat{\bm{h}}_i \in \mathbb{R}^{d}$ with a readout function on all the hidden representations resulted in each layer:
\begin{equation}
	\begin{split}
		\widehat{\bm{h}}_u &= \operatorname{READOUT} (\widehat{\bm{h}}_u^0, \widehat{\bm{h}}_u^1, \cdots, \widehat{\bm{h}}_u^{L_{ui}}), \\
		\widehat{\bm{h}}_i &= \operatorname{READOUT} (\widehat{\bm{h}}_i^0, \widehat{\bm{h}}_i^1, \cdots, \widehat{\bm{h}}_i^{L_{ui}}),
	\end{split}
	\label{eq:ui_emb}
\end{equation}
where the \(\operatorname{READOUT} \) function can be any differentiable function, $\widehat{\bm{h}}_u^0$ and $\widehat{\bm{h}}_{i}^{0} = \widetilde{\bm{h}}_{i}^{0}$ denotes the ID embeddings of user $u$ and item $i$, respectively. We use the default mean function of LightGCN~\cite{he2020lightgcn} for embedding readout.

Finally, we use the user representation output by the user-item graph as its final representation.
For the item, we sum up the representations obtained from the two graphs as its final representation. 
\begin{equation}
	\begin{split}
		{\bm{h}_u} &= \widehat{\bm{h}}_u, \\
		{\bm{h}_i} &= \widetilde{\bm{h}}_i + \widehat{\bm{h}}_i.
	\end{split}
\end{equation}
To fully explore the raw features, we project multimodal features of item $i$ in each modality via MLPs. 
\begin{equation}
	\bm{h}_i^m = \bm{x}_i^m \bm{W}_m  + \bm{b}_m,
\end{equation}
where $\bm{W}_m \in \mathbb{R}^{d_m \times d}, \bm{b}_m \in \mathbb{R}^d$ denote the linear transformation matrix and bias in the MLP.
In this way, each uni-modal representation $\bm{h}_i^m$ shares the same latent space with its ID embedding $\bm{h}_i$.

For model optimization, we adopt the pairwise Bayesian personalized ranking (BPR) loss~\cite{rendle2009bpr}, which encourages the prediction of a positive user-item pair to be scored higher than its negative pair:
\begin{align}
	\mathcal{L}_{bpr} = \sum_{(u,i,j)\in \mathcal{D}} \Bigl( & -\operatorname{log} \sigma(\bm{h}_u^\top \bm{h}_i - \bm{h}_u^\top \bm{h}_j) + \nonumber \\ 
	& \lambda  \sum_{m \in \mathcal{M}} -\operatorname{log} \sigma(\bm{h}_u^\top \bm{h}_i^m - \bm{h}_u^\top \bm{h}_j^m) \Bigr),
	\label{eq:loss}
\end{align}
where $\mathcal{D}$ is the set of training instances, and each triple $(u,i,j)$ satisfies $\bm{A}_{ui} = 1 ~\text{and}~ \bm{A}_{uj} = 0$. $\sigma(\cdot)$ is the sigmoid function and \(\lambda\) is a hyperparameter of \fm{} to weigh the reconstruction losses between user-item ID embeddings and projected multimodal features.

\subsection{Top-$K$ Recommendation}
To generate item recommendations for a user, we first predict the interaction scores between the user and candidate items. Then, we rank candidate items based on the predicted interaction scores in descending order, and choose $K$ top-ranked items as recommendations to the user. 
The interaction score is calculated as:
\begin{equation}
	r(\bm{h}_u, \bm{h}_i) = \bm{h}_u^\top \bm{h}_i.
\end{equation}
A high score suggests that the user prefers the item.
Note that we only use user and item ID embeddings for prediction, because we empirically find the adding of projected item multimodal features in prediction does not show improvement on performance.
However, the item multimodal representations can partially benefit the learning of user representations in \fm{} via Eq.~\eqref{eq:loss}.

\section{Spectral Analysis}
In this section, we examine the benefits of freezing the item-item graph in \fm{} through the lens of spectral analysis.
We show that it possesses a tighter upper bound on the graph spectrum.
In addition, we empirically calculate the largest eigenvalues of the normalized graph Laplacian on experimental datasets to validate our analysis.

\subsection{Preliminaries}
Without loss of generality, we analyze the normalized item-item adjacency matrix $\widetilde{\bm{S}}^m = ({\bm{D}^m})^{-\frac{1}{2}} \widehat{\bm{S}}^m ({\bm{D}^m})^{-\frac{1}{2}}$ within a specific modality $m$. We denote the corresponding matrix of LATTICE as $\widetilde{\bm{S}'}^m = ({\bm{D}'^m})^{-\frac{1}{2}} \widehat{\bm{S}'}^m ({\bm{D}'^m})^{-\frac{1}{2}}$. We omit $m$ for clarity in the following analysis.

\subsection{Theoretical Analysis}

\begin{lemma}
The eigenvalues of the frozen matrix in \fm{} possess a tighter upper bound.
\end{lemma}
\begin{proof}
	Let $\lambda$ be an eigenvalue of the non-negative matrix $\widetilde{\bm{S}}$. 
	\begin{align}
	\widetilde{\bm{S}} \bm{x}=\lambda \bm{x} & \Longrightarrow \lambda\left|x_i\right|=\left|\sum_{j=1}^n \tilde{s}_{i j} x_j\right| \leq\left(\sum_{j=1}^n \tilde{s}_{i j}\right) \max _j\left|x_j\right| \nonumber \\
	& \Longrightarrow \lambda \max _i\left|x_i\right| \leq \max _i\left(\sum_{j=1}^n \tilde{s}_{i j}\right) \max _j\left|x_j\right| \nonumber\\
	& \Longrightarrow \lambda \leq \max _i\left(\sum_{j=1}^n \tilde{s}_{i j}\right) \leq n \max _{i,j} \tilde{s}_{i j}
	\label{eq:max_row}
	\end{align}
	Each element in $\widetilde{\bm{S}}$ of \fm{} and $\widetilde{\bm{S}'}$ of LATTICE is positive, with the form of $\tilde{s}_{ij}/\sqrt{d_{ii} d_{jj}}$ and $\tilde{s'}_{ij}/\sqrt{d'_{ii} d'_{jj}}$. 
	As function $g(\tilde{s'}_{ij}) = \tilde{s'}_{ij}/(\tilde{s'}_{ij} +\sum_{k\neq i} \tilde{s'}_{kj})$ is monotonically increasing, and maximizing Eq.~\eqref{eq:max_row} requires the maximized $\tilde{s'}_{ij} \geq \bar{d_{ii}}$ (average value of node degree). Hence, we can derive $\max _{i,j} \tilde{s}_{i j} \leq \max _{i,j} \tilde{s'}_{i j}$. 
\end{proof}
In consideration of the page limit, we have moved the more detailed proof to Appendix~\ref{append:lemma}. Readers can access it for a more comprehensive understanding of our work. 

\begin{table}[]
	\caption{Comparison of the largest eigenvalues of \fm{} and LATTICE on different datasets (lower is better).}
	\vspace{-8pt}
	\begin{tabular}{p{1.8cm}lll}
		\toprule
		Model   & Baby   & Sports & Clothing \\
		\midrule
		\fm{} & 1.1685 & 1.1016 & 1.0932   \\
		LATTICE   & 1.1796 & 1.1180 & 1.1210  \\
		\bottomrule
	\end{tabular}
	\vspace{-8pt}
	\label{tab:egien}	
\end{table}

\subsection{Empirical Study}
We compute the largest eigenvalues of the item-item matrix in both \fm{} and LATTICE, as shown in Table~\ref{tab:egien}. The results validate our previous analysis. A tighter upper bound of eigenvalues in the item-item graph ensures that the frozen item-item graph in FREEDOM can act as a low-pass filter, eliminating the effect of negative coefficients at large frequencies. We further investigate the impact of the frozen item-item graph in \fm{} and the learnable graph in LATTICE on validation accuracy. The plots in Fig.~\ref{fig:val-all} reveal that graphs in FREEDOM trained with this property achieve higher accuracy and more stable performance than LATTICE. 

\begin{figure}[t!]
	\centering
	\subfloat[Baby]{\includegraphics[width=0.46\linewidth]{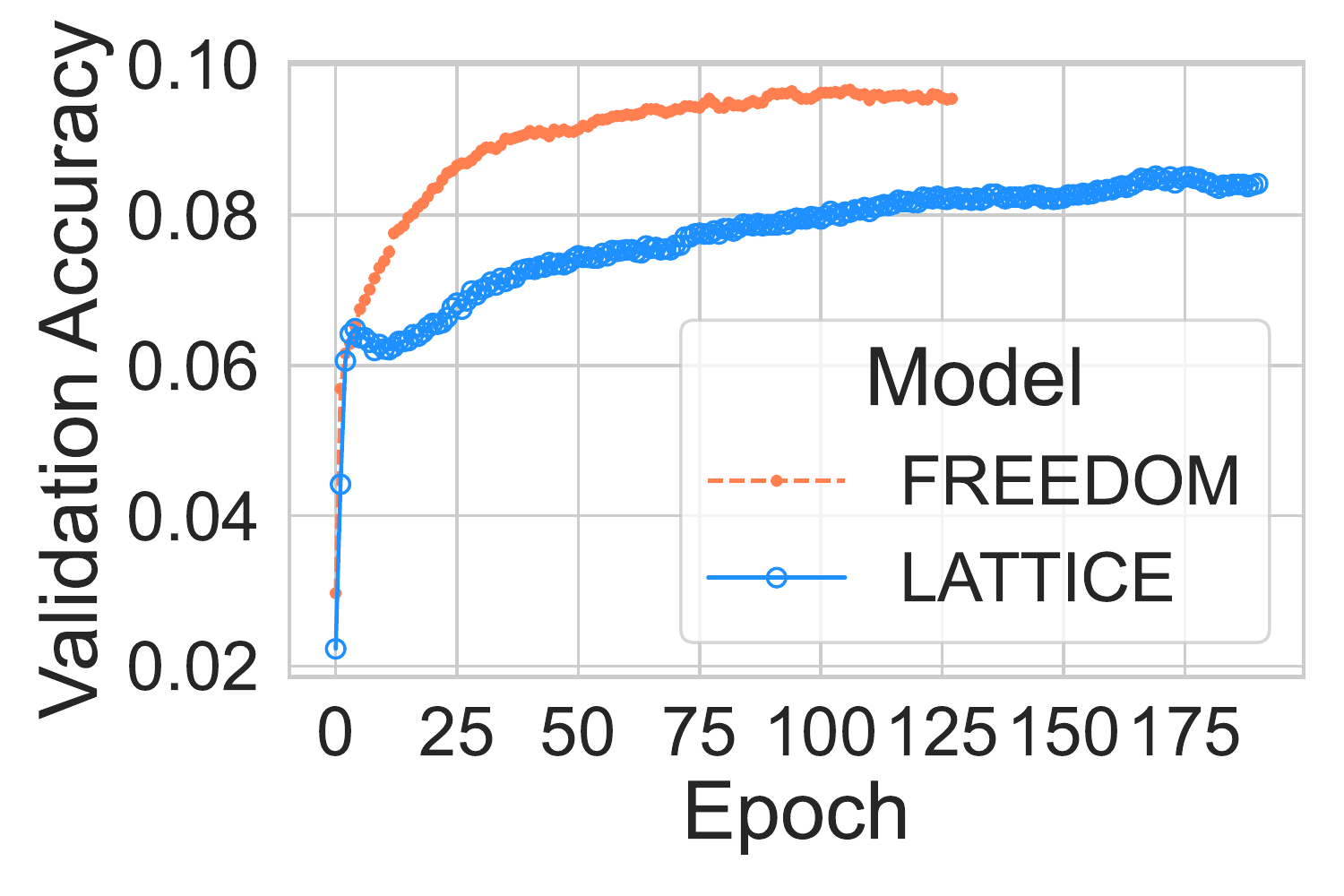}} \hspace{0.1cm}
	\subfloat[Clothing]{\includegraphics[width=0.46\linewidth]{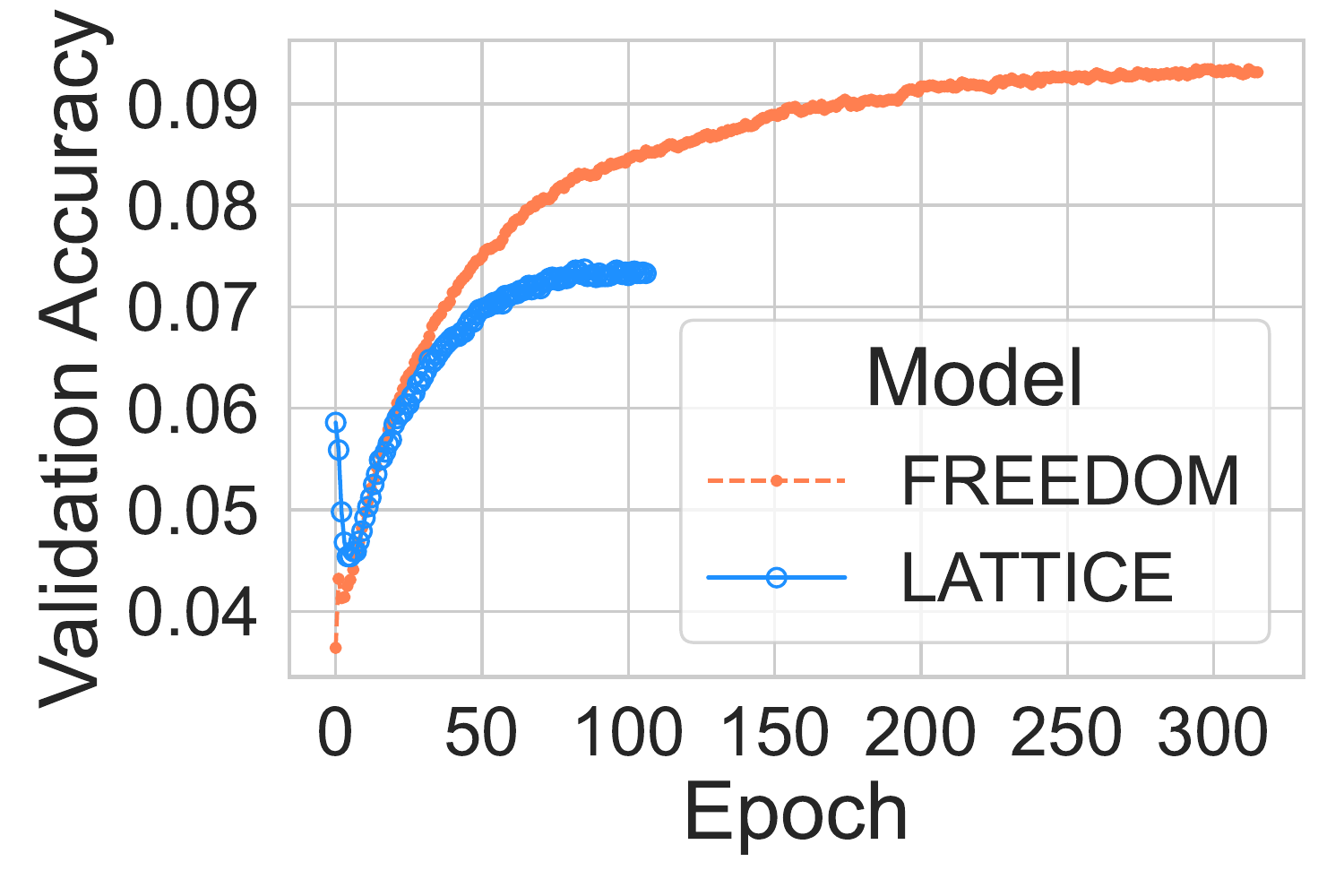}}
	\vspace{-3pt}
	\caption{Validation accuracy (\ie Recall@20) of \fm{} and LATTICE using frozen and learnable item-item matrices, respectively.}
	\label{fig:val-all}
	\vspace{-3pt}
\end{figure}

\section{Experiments}
\label{sec:experiments}
To evaluate the effectiveness and efficiency of our proposed \fm{}, we conduct experiments to answer the following research questions on three real-world datasets.
\begin{itemize}
	\item{\textbf{RQ1}}: How does \fm{} perform compared with the state-of-the-art methods for recommendation? As our model improves LATTICE~\cite{zhang2021mining} by freezing and denoising the graph structures, how about its improvement over LATTICE?
	\item{\textbf{RQ2}}: How efficient of our proposed \fm{} in terms of computational complexity and memory cost?
	\item{\textbf{RQ3}}: How do different components in \fm{} influence its recommendation accuracy? 
	\item{\textbf{RQ4}}: How sensitive is our model under the perturbation of hyperparameters?
\end{itemize}

\begin{table}[bpt]
	\centering	
	\def\arraystretch{0.9}	
	\caption{Statistics of the experimental datasets.}
	\begin{tabular}{l r r r r}
		\toprule
		Dataset & \# Users & \# Items & \# Interactions & Sparsity \\
		\midrule
		Baby & 19,445 & 7,050 & 160,792 & 99.88\% \\
		Sports & 35,598 & 18,357 & 296,337 & 99.95\%\\
		Clothing & 39,387 & 23,033 & 278,677 & 99.97\%\\
		\bottomrule
	\end{tabular}
	\vspace{-12pt}
	\label{tab:datasets}
\end{table}

\subsection{Experimental Datasets}
\label{sec:exp_data}
Following~\cite{zhang2021mining}, we conduct experiments on three categories of the Amazon review dataset~\cite{he2016ups}: (a) \emph{Baby}, (b) \emph{Sports and Outdoors}, and (c) \emph{Clothing, Shoes and Jewelry}. For simplicity, we denote them as \emph{Baby}, \emph{Sports} and \emph{Clothing}, respectively.
The Amazon review dataset provides both visual and textual information about the items and varies in the number of items under different categories. 
The raw data of each dataset are pre-processed with a 5-core setting on both items and users, and their filtered results are presented in Table~\ref{tab:datasets}.
We directly use the 4,096-dimensional visual features extracted by pre-trained Convolutional Neural Networks~\cite{he2016ups}.
For the textual modality, we extract a 384-dimensional textual embedding by utilizing sentence-transformers~\cite{reimers2019sentence}~\footnote{We use \textit{all-MiniLM-L6-v2} because it is much faster and offers better quality than the pre-trained model used in LATTICE. Except for the pre-trained model, our experimental setting completely follows LATTICE~\cite{zhang2021mining}.} on the concatenation of the title, descriptions, categories, and brand of each item. 

\subsection{Baseline Methods}
To demonstrate the effectiveness of \fm{}, we compare it with the state-of-the-art recommendation methods in two categories.

First category comprises two general CF models that recommend personalized items to users solely based on user-item interactions:
\begin{itemize}
	\item \textbf{BPR}~\cite{rendle2009bpr} optimizes the latent representations of users and items under the framework of matrix factorization (MF) with a BPR loss.
	\item \textbf{LightGCN}~\cite{he2020lightgcn} simplifies the vanilla GCN by removing its non-linear activation and feature transformation layers for recommendation.
\end{itemize}

While the second line of work includes six multimodal recommendation models that leverage both user interactions and multimodal information about items for recommendation:
\begin{itemize}
	\item \textbf{VBPR}~\cite{he2016vbpr} incorporates visual features for user preference learning under the MF framework and BPR loss. Following~\cite{zhang2021mining, wang2021dualgnn}, we concatenate the multimodal features of an item as its visual feature for user preference learning.
	\item \textbf{MMGCN}~\cite{wei2019mmgcn} fuses the representations generated by GCNs in each modality of items for recommendation.
	\item \textbf{GRCN}~\cite{wei2020graph} refines the user-item bipartite graph with the removal of false-positive edges for multimodal recommendation. 
	Based on the refined graph, it then learns user and item representations by performing information propagation and aggregation via GCNs.
	\item \textbf{DualGNN}~\cite{wang2021dualgnn} augments the representations of users in GCNs with an additional user-user correlation graph that is extracted from the user-item graph.
	\item \textbf{LATTICE}~\cite{zhang2021mining} learns the latent semantic item-item structures from the multimodal features for recommendation. 
	As \fm{} improves LATTICE, we highlight its improvement over LATTICE in Table~\ref{tab:perform}.
	\item \textbf{SLMRec}~\cite{tao2022self} incorporates self-supervised learning into multimedia recommendation. It proposes three data augmentations to uncover the multimodal patterns in data for contrastive learning.
\end{itemize}

\subsection{Evaluation Protocols}
For a fair comparison, we adopt the same evaluation settings of~\cite{wang2021dualgnn,zhang2021mining, tao2022self}. Specifically, we use two widely-used evaluation protocols for top-$K$ recommendation: Recall@$K$ and NDCG@$K$, which we refer to as R@$K$ and N@$K$ in brief. We report the average metrics of all users in the test set under both $K=10$ and $K=20$.
For each user in the evaluated dataset, we randomly split 80\% of historical interactions for training, 10\% for validation and the remaining 10\% for testing.
During training, we conduct the negative sampling strategy to pair each observed user-item interaction in the training set with one negative item that the user does not interact with before.
We use the all-ranking protocol to compute the evaluation metrics for recommendation accuracy comparison. 

\subsection{Implementation and Hyperparameter Settings}
Following existing work~\cite{he2020lightgcn, zhang2021mining}, we fix the embedding size of both users and items to 64 for all models, initialize the embedding parameters with the Xavier method~\cite{glorot2010understanding}, and use Adam~\cite{kingma2015adam} as the optimizer. For a fair comparison, we carefully tune the parameters of each model following their published papers. All models are implemented by PyTorch~\cite{paszke2019pytorch} and evaluated on a Tesla V100 GPU card with 32 GB memory.
To reduce the hyperparameters searching space of \fm{}, we fix the number of GCN layers in the user-item bipartite graph and the item-item graph at $L_{ui}=2$ and $L_{ii}=1$, respectively. 
We empirically fix the hyperparameter of $\lambda$ at $1e-03$ and the visual feature ratio $\alpha_v$ at 0.1.
We then perform a grid search on other hyperparameters of \fm{} across all datasets to conform to its optimal settings. Specifically, the ratio $\rho$ of the degree-sensitive edge pruning is searched from \{0.8, 0.9\}.
For convergence consideration, the early stopping and total epochs are fixed at 20 and 1000, respectively. 
Following~\cite{zhang2021mining}, we use R@20 on the validation data as the training stopping indicator.
To ensure a fair comparison, all baseline models as well as our proposed model have been integrated into the unified multimodal recommendation framework, MMRec~\cite{zhou2023mmrec}.

\begin{table*}[bpt]
	\small
	\centering	
	\def\arraystretch{1.1}	
	\caption{Overall performance achieved by different recommendation methods in terms of Recall and NDCG.  We mark the global best results on each dataset under each metric in \textbf{boldface} and the second best is \underline{underlined}. Improvement percentage (\emph{improv.}) is calculated as the ratio of performance increment from LATTICE to \fm{} under each dataset and metric. To verify the stability of our method, we conduct experiments across 5 different seeds and state the improvements are statistically significant at the level of $p$ < 0.01 with a paired $t$-test.}
	\begin{tabular}{llcccccccccc}
		\toprule
		\multirow{2}{*}{Dataset} & \multirow{2}{*}{Metric} & \multicolumn{2}{c}{General CF model} & \multicolumn{7}{c}{Multimodal model} & \\
		\cmidrule(lr){3-4} \cmidrule(lr){5-12}
		& & BPR & LightGCN & VBPR & MMGCN & GRCN & DualGNN & SLMRec & LATTICE & \fm{} & \emph{improv.}\\
		\midrule
		\multirow{4}{*}{Baby} & R@10 & 0.0357 & 0.0479 & 0.0423 & 0.0421 & 0.0532 & 0.0513 & 0.0521 & \underline{0.0547} & \textbf{0.0627} & 14.63\% \\
		& R@20 & 0.0575 & 0.0754 & 0.0663 & 0.0660 & 0.0824 & 0.0803 & 0.0772 & \underline{0.0850} & \textbf{0.0992} & 16.71\% \\
		& N@10 & 0.0192 & 0.0257 & 0.0223 & 0.0220 & 0.0282 & 0.0278 & 0.0289 & \underline{0.0292} & \textbf{0.0330}  & 13.01\% \\
		& N@20 & 0.0249 & 0.0328 & 0.0284 & 0.0282 & 0.0358 & 0.0352 & 0.0354 & \underline{0.0370} & \textbf{0.0424}  & 14.59\% \\
		\midrule
		\multirow{4}{*}{Sports} & R@10 & 0.0432 & 0.0569 & 0.0558 & 0.0401 & 0.0599 & 0.0588 & \underline{0.0663} & 0.0620 & \textbf{0.0717} & 15.65\% \\
		& R@20 & 0.0653 & 0.0864 & 0.0856 & 0.0636 & 0.0919 & 0.0899 & \underline{0.0990} & 0.0953 & \textbf{0.1089}  & 14.27\% \\
		& N@10 & 0.0241 & 0.0311 & 0.0307 & 0.0209 & 0.0330 & 0.0324 & \underline{0.0365} & 0.0335 & \textbf{0.0385}  & 14.93\% \\
		& N@20 & 0.0298 & 0.0387 & 0.0384 & 0.0270 & 0.0413 & 0.0404 & \underline{0.0450} & 0.0421 & \textbf{0.0481}  & 14.25\% \\
		\midrule
		\multirow{4}{*}{Clothing} & R@10 & 0.0206 & 0.0361 & 0.0281 & 0.0227 & 0.0421 & 0.0452 & 0.0442 & \underline{0.0492} & \textbf{0.0629} & 27.85\% \\
		& R@20 & 0.0303 & 0.0544 & 0.0415 & 0.0361 & 0.0657 & 0.0675 & 0.0659 & \underline{0.0733} & \textbf{0.0941}  & 28.38\% \\
		& N@10 & 0.0114 & 0.0197 & 0.0158 & 0.0120 & 0.0224 & 0.0242 & 0.0241 & \underline{0.0268} & \textbf{0.0341}  &  27.24\% \\
		& N@20 & 0.0138 & 0.0243 & 0.0192 & 0.0154 & 0.0284 & 0.0298 & 0.0296 & \underline{0.0330} & \textbf{0.0420}  & 27.27\% \\
		\bottomrule	
	\end{tabular}
	\label{tab:perform}	
\end{table*}

\subsection{Performance Comparison}
\paragraph[Effectiveness of \fm{}]{Effectiveness (\textbf{RQ1})}
Table~\ref{tab:perform} reports the comparison of recommendation accuracy in terms of Recall and NDCG, from which we have the following observations: 
1). Although LATTICE obtains the second best results on Baby and Clothing datasets, \fm{} improves LATTICE by an average of 19.07\% across all datasets and outperforms other baselines as well.
The improvements attribute to the graph structures denoising and freezing components of \fm{}. 
Denoising the user-item graph in $\bm{A}$ reduces the impact of noise signal in false-positive interactions.
The frozen item-item graph in \fm{} ensures the items are relevant to each other if linked because it built on the raw multimodal features. 
On the contrary, the latent item-item graph is dynamically learned via projecting the raw multimodal features into a low-dimensional space. 
The affinity of two items in the graph depends on not only the raw multimodal features but also the projectors (\ie MLPs). 
2). Generally, both general MF (\ie BPR) and graph-based (\ie LightGCN) models can benefit from the multimodal information.
For example, VBPR leveraging multimodal features outperforms BPR by up to 27.63\% on average. 
Most graph-based multimodal models except MMGCN show better recommendation results than MF models (\ie BPR and VBPR) and general CF models.
However, we observe that the performance of MMGCN is inferior to that of BPR on Sports dataset. One potential reason for this might be that the textual and visual representations extracted from the texts and images of products in the Sports dataset are less informative than those of the Baby and Clothing datasets. Another reason may be the fusion mechanism used in MMGCN, which sums up the features from all modalities, making it difficult to discern the prominent contribution of individual modalities. 
Our proposed \fm{} utilizing LightGCN as its backbone network gains a significant margin of 42.79\% over the original LightGCN on average.
3). In graph-based multimodal models, compared with MMGCN, GRCN utilizes simplified graph convolutional operations and refines the user-item bipartite graph by assigning weights based on the affinities of user preferences and item contents. It shows better performance in dense datasets, but is inferior in Clothing dataset. The reason is that the assignments of low weights on less similar user-item edges impede the information propagation of GCNs in a sparse dataset. 
DualGNN and SLMRec utilize auxiliary information to augment the representation learnings on users and items. They have comparable performance with each other on the evaluated dataset except Sports. 
One potential reason is that the extracted textual and visual representations from the texts and images of products in Sports are less informative than that of Baby and Clothing datasets.
The evidence is that in Table~\ref{tab:perform}, the improvements achieved by GRCN and DualGNN over LightGCN in Sports are not as significant as that of Baby and Clothing.
SLMRec augments the representation from each modality of an item with self-supervised learning and gains better results than GRCN and DualGNN in Sports.
Analogously, both LATTICE and \fm{} augment the representations of items via a latent item-item graph and show competitive results.

Overall, the experiment results in Table~\ref{tab:perform} validate the effectiveness of augmenting item representation with multimodal information. Furthermore, freezing the latent graph structure of the item as in \fm{} can help achieve even better performance.
\paragraph[Efficiency of \fm{}]{Efficiency (\textbf{RQ2})}
We report the memory and training time consumed by \fm{} and baselines in Table~\ref{tab:perform_eff}.
From the table, we observe: 1). Multimodal models usually utilize more memory than the general CF models as they need to process modality-aware features. Graph-based models perform convolutional operations on graphs and cost more time on training.
DualGNN consumes even more time in training because the convolutional operations are performed not only on the user-item graph but also on the user-user relationship graph.
MMGCN propagates convolutional operators on each modality, resulting in increased training time.
2). Compared \fm{} with LATTICE, the proposed \fm{} can reduce the memory cost and training time of LATTICE on Clothing by 6$\times$ and 4$\times$, respectively.
\fm{} removes the construction of latent item-item graph in each training epoch, it pre-builds the graph before training and freezes it during training.
With the removal of item-item graph construction in training, \fm{} is preferable with large graphs.

\begin{table*}[bpt]
	\small
	\centering	
	\caption{Comparison of \fm{} against state-of-the-art baselines on model efficiency.}
	\begin{tabular}{llccccccccc}
		\toprule
		\multirow{2}{*}{Dataset} & \multirow{2}{*}{Metric} & \multicolumn{2}{c}{General CF model} & \multicolumn{7}{c}{Multimodal model} \\
		\cmidrule(lr){3-4} \cmidrule(lr){5-11}
		& & BPR & LightGCN & VBPR & MMGCN & GRCN & DualGNN & SLMRec & LATTICE & \fm{} \\
		\midrule
		\multirow{2}{*}{Baby} & Memory (GB) & 1.59 & 1.69 & 1.89 & 2.69 & 2.95 & 2.05 & 2.08 & 4.53  &  2.13 \\
		& Time ($s$/epoch) & 0.47 & 0.99 & 0.57 & 3.48 & 2.36 & 7.81 & 1.91 & 1.61 & 1.25 \\
		\midrule
		\multirow{2}{*}{Sports} & Memory (GB) & 2.00 & 2.24 & 2.71 & 3.91 & 4.49 & 2.81 & 3.04 & 19.93 & 3.34 \\
		& Time ($s$/epoch) & 0.95 & 2.86 & 1.28 & 16.60 & 6.74 & 12.60  & 5.52 & 10.71 & 3.46 \\
		\midrule
		\multirow{2}{*}{Clothing} & Memory (GB) & 2.16 & 2.43 & 3.02 & 4.24 & 4.58 & 3.02 & 3.40 & 28.22 & 4.15 \\
		& Time ($s$/epoch) & 0.97 & 2.94 & 1.29 & 16.71 & 7.21 & 13.44 & 5.09 & 15.96 & 3.67 \\
		\bottomrule
	\end{tabular}
	\label{tab:perform_eff}	
\end{table*}

\subsection{Ablation Study (\textbf{RQ3})}
In this section, we decouple the proposed \fm{} and evaluate the contribution of each component with regard to recommendation accuracy.
We design the following variants of \fm{} based on its architecture: 
\begin{itemize}
	\item \textbf{\fm{}-D} denoises the user-item graph in \fm{} without freezing the item-item graph.
	\item \textbf{\fm{}-F} improves LATTICE merely by the freezing component of \fm{}, which freezes the item-item graph as introduced in Section~\ref{sec:freedom}.
	\item \textbf{\fm{}-R} replaces the denoising method of degree-sensitive edge pruning in \fm{} with random edge dropout~\cite{rong2020dropedge}. 
	\item \textbf{\fm{}-0} ablates the contribution of multimodal losses in Eq.~\eqref{eq:loss} by setting $\lambda$ to 0. 
\end{itemize}
We report the comparison results in terms of R@20 and N@20 with Fig.~\ref{fig:ab-comp}.
The results show that the freezing component of \fm{} attributes significantly to its overall performance.
However, denoising the user-item graph can further improve the performance.
Freezing the item-item graph in \fm{} (\ie \fm{}-F) gains consistent improvements over LATTICE on three datasets, but not for the denoising component. 
We speculate that without the freezing component, the denoising component may affect graph connectivity on sparse graphs and result in performance degradation. 
\fm{}-R with random edge dropout shows slight improvement over \fm{}-F but is still worse than \fm{}, showing the effectiveness of degree-sensitive edge pruning.
The performance of \fm{}-0 is comparable to \fm{}, indicating the effective design of \fm{} in freezing and denoising graph structures for recommendation. However, \fm{} can benefit slightly from the inclusion of multimodal losses.

\label{sec:ab}
\begin{figure}[bpt]
	\centering
	\begin{minipage}{0.46\textwidth}
		\includegraphics[width=\columnwidth, trim={7 0 7 0},clip]{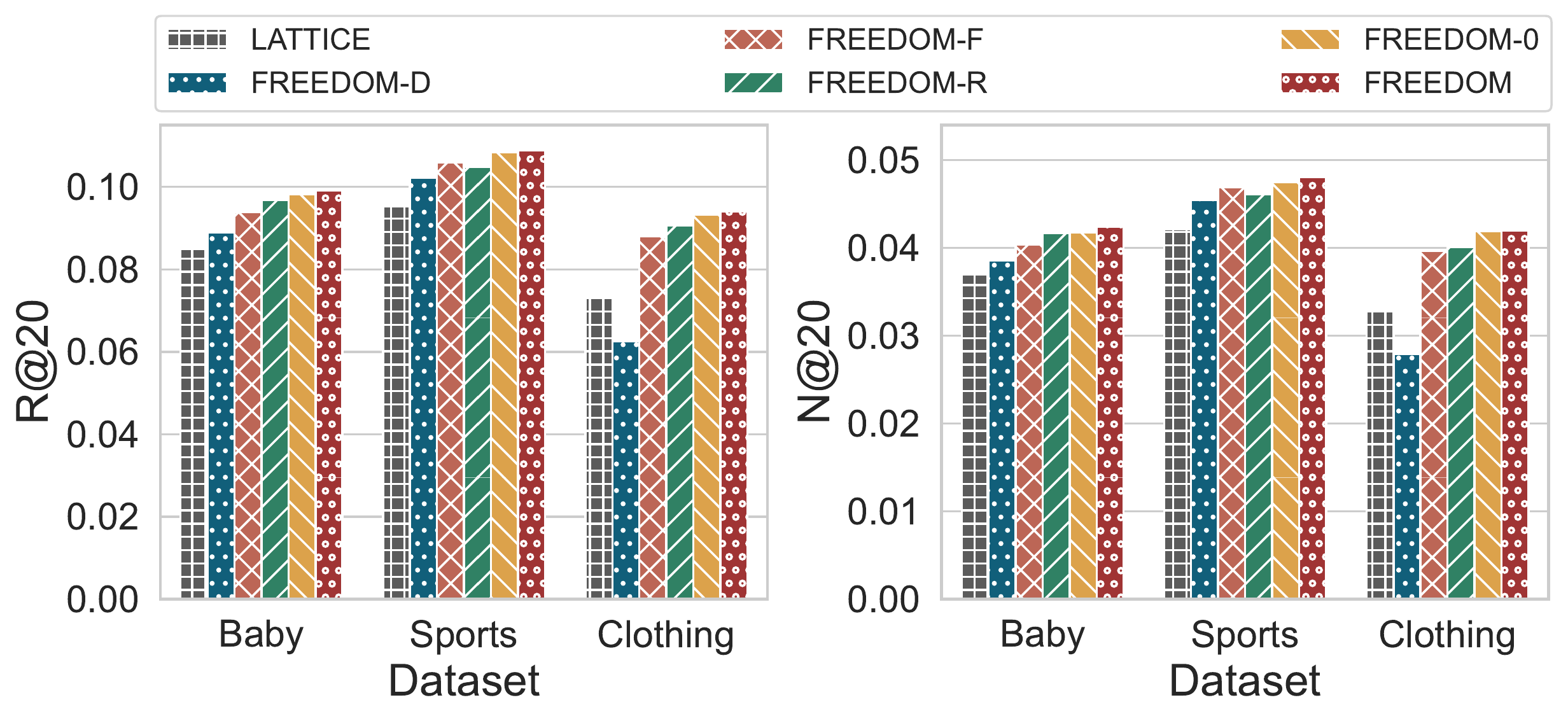} 
		\vspace{-14pt}
		\caption{Comparison of LATTICE with \fm{} variants.}
		\label{fig:ab-comp}
	\end{minipage}\hfill
	\begin{minipage}{0.46\textwidth}
		\includegraphics[width=\columnwidth, trim={15 0 15 0},clip]{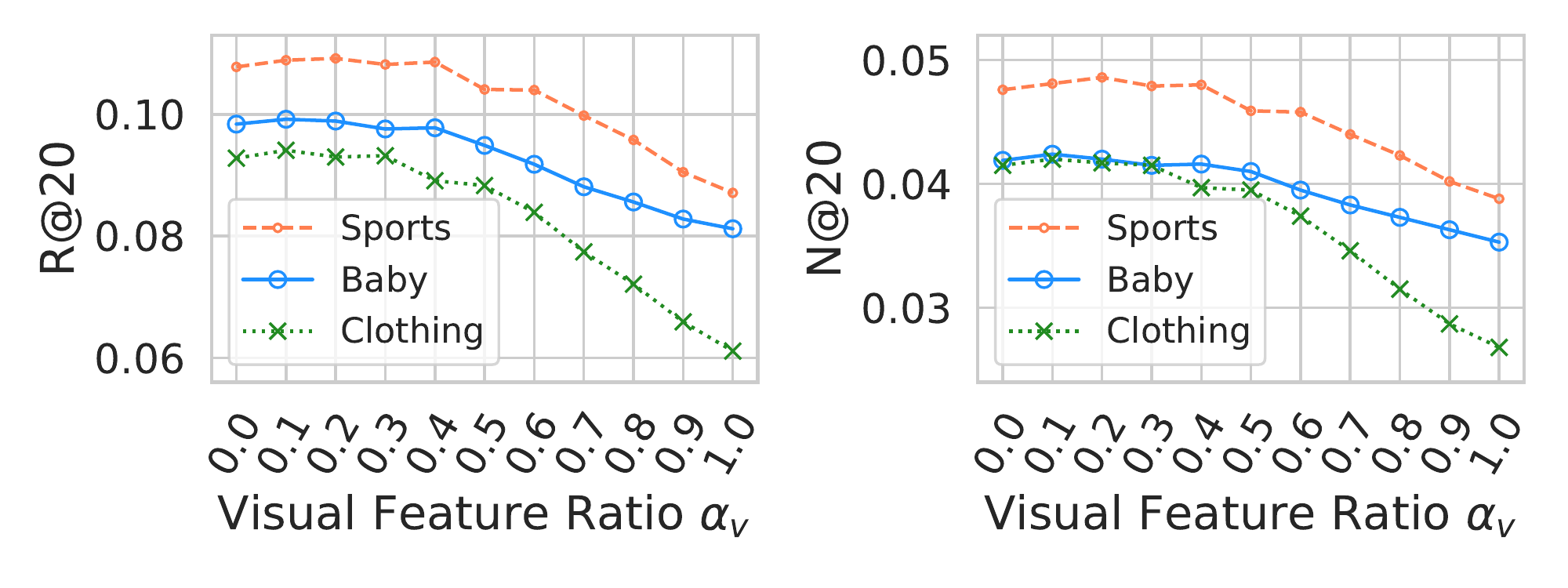} 
		\vspace{-15pt}
		\caption{Performance of \fm{} changes with the ratio of visual features in constructing the item-item graph.}
		\label{fig:ab-vt}
	\end{minipage}
\end{figure}

\begin{figure}[bpt]
	\centering
	\begin{subfigure}[b]{0.48\textwidth}
		\includegraphics[width=1\linewidth]{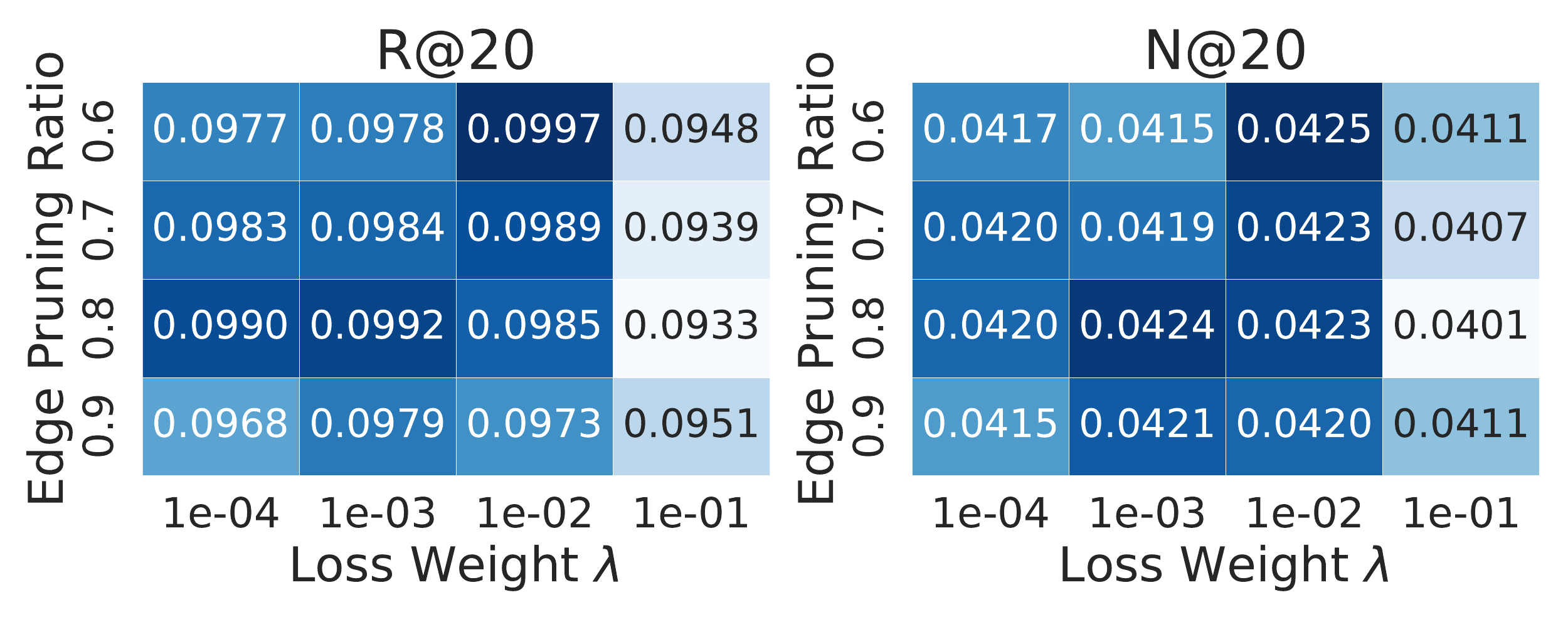}
		\vspace{-13pt}
		\caption{Baby}
		\label{fig:ht-baby20} 
	\end{subfigure}
	\hfill
	\begin{subfigure}[b]{0.48\textwidth}
		\includegraphics[width=1\linewidth]{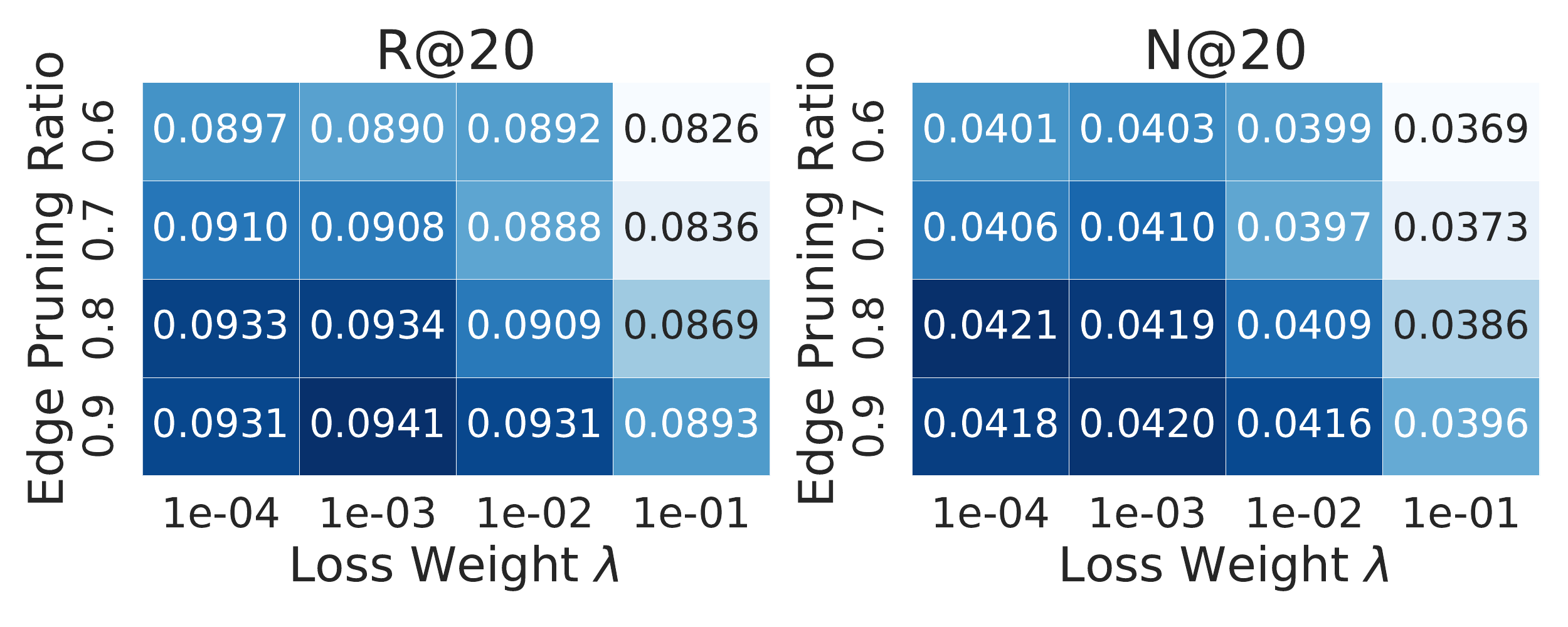}
		\vspace{-13pt}
		\caption{Clothing}
		\label{fig:ht-clothing20}
	\end{subfigure}
	\caption[]{Performance of \fm{} with regard to different loss weights $\lambda$ and edge pruning ratios $\rho$ on Baby and Clothing datasets.}
\end{figure}

\subsection{Hyperparameter Sensitivity Study (\textbf{RQ4})}
\textbf{Multimodal Features.} \fm{} uses the raw textual and visual features extracted from pre-trained models to construct the latent item-item graph. 
We first study how information from different modalities could affect the performance of \fm{} by adjusting the ratio of visual features from 0.0 to 1.0 with a step of 0.1 in constructing the graph.
A ratio of 0.0 on visual features means that the construction of item-item graph depends only on textual features. 
The item-item graph, on the other hand, is based on visual features if the ratio is 1.0.
The recommendation results on R@20 and N@20 for three evaluated datasets are shown in Fig.~\ref{fig:ab-vt}.
From the results, we may infer that the textual features are more informative than the visual features in constructing an effective item-item graph.

\textbf{The Dropout Ratio and Loss Weight.} 
We vary the edge pruning ratio $\rho$ in the denoising component of \fm{} from 0.6 to 0.9 with a step of 0.1, and vary the loss weight $\lambda$ in \{1e-04, 1e-03, 1e-02, 1e-01\}. Fig.~\ref{fig:ht-baby20} and~\ref{fig:ht-clothing20} show the performance achieved by \fm{} under different combinations of embedding dropout ratios and weights on Baby and Clothing, respectively. 
The results suggest a high edge pruning ratio of \fm{} on large graphs.
Compared with the edge pruning ratio, the performance of \fm{} is less sensitive to the settings of loss trade-off $\lambda$. 
However, placing a high weight on multimodal loss might limit the expressive power of \fm{} and results in performance degradation.

\section{Conclusion}
In this paper, we experimentally reveal that the graph structure learning in a state-of-the-art multimodal recommendation model (\ie LATTICE~\cite{zhang2021mining}) plays a trivial role in its performance.
It is the item-item graph constructed from raw multimodal features that contributes to the recommendation accuracy.
Based on the finding, we propose a model that freezes the item-item graph and denoises the user-item graph simultaneously for multimodal recommendation.
Through both theoretical and empirical analysis, we demonstrate that freezing the item-item graph in \fm{} can yield various benefits.
In denoising, we devise a degree-sensitive edge pruning method to sample the user-item graph, which shows better performance than the random edge dropout~\cite{rong2020dropedge} for recommendation.
Finally, we conduct extensive experiments to demonstrate the effectiveness and efficiency of \fm{} in multimodal recommendation.

\begin{acks}
	This work was supported by Alibaba Group through Alibaba Innovative Research (AIR) Program and Alibaba-NTU Singapore Joint Research Institute (JRI), Nanyang Technological University, Singapore.
\end{acks}

\bibliographystyle{ACM-Reference-Format}
\bibliography{reference}


\begin{thebibliography}{36}


\ifx \showCODEN    \undefined \def \showCODEN     #1{\unskip}     \fi
\ifx \showDOI      \undefined \def \showDOI       #1{#1}\fi
\ifx \showISBNx    \undefined \def \showISBNx     #1{\unskip}     \fi
\ifx \showISBNxiii \undefined \def \showISBNxiii  #1{\unskip}     \fi
\ifx \showISSN     \undefined \def \showISSN      #1{\unskip}     \fi
\ifx \showLCCN     \undefined \def \showLCCN      #1{\unskip}     \fi
\ifx \shownote     \undefined \def \shownote      #1{#1}          \fi
\ifx \showarticletitle \undefined \def \showarticletitle #1{#1}   \fi
\ifx \showURL      \undefined \def \showURL       {\relax}        \fi
\providecommand\bibfield[2]{#2}
\providecommand\bibinfo[2]{#2}
\providecommand\natexlab[1]{#1}
\providecommand\showeprint[2][]{arXiv:#2}

\bibitem[Chen et~al\mbox{.}(2009)]%
        {chen2009fast}
\bibfield{author}{\bibinfo{person}{Jie Chen}, \bibinfo{person}{Haw-ren Fang},
  {and} \bibinfo{person}{Yousef Saad}.} \bibinfo{year}{2009}\natexlab{}.
\newblock \showarticletitle{Fast Approximate kNN Graph Construction for High
  Dimensional Data via Recursive Lanczos Bisection.}
\newblock \bibinfo{journal}{\emph{Journal of Machine Learning Research}}
  \bibinfo{volume}{10}, \bibinfo{number}{9} (\bibinfo{year}{2009}).
\newblock


\bibitem[Chen et~al\mbox{.}(2017)]%
        {chen2017attentive}
\bibfield{author}{\bibinfo{person}{Jingyuan Chen}, \bibinfo{person}{Hanwang
  Zhang}, \bibinfo{person}{Xiangnan He}, \bibinfo{person}{Liqiang Nie},
  \bibinfo{person}{Wei Liu}, {and} \bibinfo{person}{Tat-Seng Chua}.}
  \bibinfo{year}{2017}\natexlab{}.
\newblock \showarticletitle{Attentive collaborative filtering: Multimedia
  recommendation with item-and component-level attention}. In
  \bibinfo{booktitle}{\emph{Proceedings of the 40th International ACM SIGIR
  conference on Research and Development in Information Retrieval}}.
  \bibinfo{pages}{335--344}.
\newblock


\bibitem[Chen et~al\mbox{.}(2020)]%
        {chen2020simple}
\bibfield{author}{\bibinfo{person}{Ming Chen}, \bibinfo{person}{Zhewei Wei},
  \bibinfo{person}{Zengfeng Huang}, \bibinfo{person}{Bolin Ding}, {and}
  \bibinfo{person}{Yaliang Li}.} \bibinfo{year}{2020}\natexlab{}.
\newblock \showarticletitle{Simple and deep graph convolutional networks}. In
  \bibinfo{booktitle}{\emph{International Conference on Machine Learning}}.
  PMLR, \bibinfo{pages}{1725--1735}.
\newblock


\bibitem[Chen et~al\mbox{.}(2019)]%
        {chen2019personalized}
\bibfield{author}{\bibinfo{person}{Xu Chen}, \bibinfo{person}{Hanxiong Chen},
  \bibinfo{person}{Hongteng Xu}, \bibinfo{person}{Yongfeng Zhang},
  \bibinfo{person}{Yixin Cao}, \bibinfo{person}{Zheng Qin}, {and}
  \bibinfo{person}{Hongyuan Zha}.} \bibinfo{year}{2019}\natexlab{}.
\newblock \showarticletitle{Personalized fashion recommendation with visual
  explanations based on multimodal attention network: Towards visually
  explainable recommendation}. In \bibinfo{booktitle}{\emph{Proceedings of the
  42nd International ACM SIGIR Conference on Research and Development in
  Information Retrieval}}. \bibinfo{pages}{765--774}.
\newblock


\bibitem[Franceschi et~al\mbox{.}(2019)]%
        {franceschi2019learning}
\bibfield{author}{\bibinfo{person}{Luca Franceschi}, \bibinfo{person}{Mathias
  Niepert}, \bibinfo{person}{Massimiliano Pontil}, {and} \bibinfo{person}{Xiao
  He}.} \bibinfo{year}{2019}\natexlab{}.
\newblock \showarticletitle{Learning discrete structures for graph neural
  networks}. In \bibinfo{booktitle}{\emph{International conference on machine
  learning}}. PMLR, \bibinfo{pages}{1972--1982}.
\newblock


\bibitem[Glorot and Bengio(2010)]%
        {glorot2010understanding}
\bibfield{author}{\bibinfo{person}{Xavier Glorot} {and} \bibinfo{person}{Yoshua
  Bengio}.} \bibinfo{year}{2010}\natexlab{}.
\newblock \showarticletitle{Understanding the difficulty of training deep
  feedforward neural networks}. In \bibinfo{booktitle}{\emph{Proceedings of the
  thirteenth international conference on artificial intelligence and
  statistics}}. JMLR Workshop and Conference Proceedings,
  \bibinfo{pages}{249--256}.
\newblock


\bibitem[Hamilton et~al\mbox{.}(2017)]%
        {hamilton2017inductive}
\bibfield{author}{\bibinfo{person}{Will Hamilton}, \bibinfo{person}{Zhitao
  Ying}, {and} \bibinfo{person}{Jure Leskovec}.}
  \bibinfo{year}{2017}\natexlab{}.
\newblock \showarticletitle{Inductive representation learning on large graphs}.
\newblock \bibinfo{journal}{\emph{Advances in neural information processing
  systems}}  \bibinfo{volume}{30} (\bibinfo{year}{2017}).
\newblock


\bibitem[He and McAuley(2016a)]%
        {he2016ups}
\bibfield{author}{\bibinfo{person}{Ruining He} {and} \bibinfo{person}{Julian
  McAuley}.} \bibinfo{year}{2016}\natexlab{a}.
\newblock \showarticletitle{Ups and downs: Modeling the visual evolution of
  fashion trends with one-class collaborative filtering}. In
  \bibinfo{booktitle}{\emph{proceedings of the 25th international conference on
  world wide web}}. \bibinfo{pages}{507--517}.
\newblock


\bibitem[He and McAuley(2016b)]%
        {he2016vbpr}
\bibfield{author}{\bibinfo{person}{Ruining He} {and} \bibinfo{person}{Julian
  McAuley}.} \bibinfo{year}{2016}\natexlab{b}.
\newblock \showarticletitle{VBPR: visual bayesian personalized ranking from
  implicit feedback}. In \bibinfo{booktitle}{\emph{Proceedings of the AAAI
  conference on artificial intelligence}}, Vol.~\bibinfo{volume}{30}.
\newblock


\bibitem[He et~al\mbox{.}(2020)]%
        {he2020lightgcn}
\bibfield{author}{\bibinfo{person}{Xiangnan He}, \bibinfo{person}{Kuan Deng},
  \bibinfo{person}{Xiang Wang}, \bibinfo{person}{Yan Li},
  \bibinfo{person}{Yongdong Zhang}, {and} \bibinfo{person}{Meng Wang}.}
  \bibinfo{year}{2020}\natexlab{}.
\newblock \showarticletitle{Lightgcn: Simplifying and powering graph
  convolution network for recommendation}. In
  \bibinfo{booktitle}{\emph{Proceedings of the 43rd International ACM SIGIR
  Conference on Research and Development in Information Retrieval}}.
  \bibinfo{pages}{639--648}.
\newblock


\bibitem[Kazi et~al\mbox{.}(2022)]%
        {kazi2022differentiable}
\bibfield{author}{\bibinfo{person}{Anees Kazi}, \bibinfo{person}{Luca Cosmo},
  \bibinfo{person}{Seyed-Ahmad Ahmadi}, \bibinfo{person}{Nassir Navab}, {and}
  \bibinfo{person}{Michael Bronstein}.} \bibinfo{year}{2022}\natexlab{}.
\newblock \showarticletitle{Differentiable graph module (dgm) for graph
  convolutional networks}.
\newblock \bibinfo{journal}{\emph{IEEE Transactions on Pattern Analysis and
  Machine Intelligence}} (\bibinfo{year}{2022}).
\newblock


\bibitem[Kingma and Ba(2015)]%
        {kingma2015adam}
\bibfield{author}{\bibinfo{person}{Diederik~P Kingma} {and}
  \bibinfo{person}{Jimmy Ba}.} \bibinfo{year}{2015}\natexlab{}.
\newblock \showarticletitle{Adam: A method for stochastic optimization}. In
  \bibinfo{booktitle}{\emph{International Conference on Learning
  Representations}}.
\newblock


\bibitem[Liu et~al\mbox{.}(2019)]%
        {liu2019user}
\bibfield{author}{\bibinfo{person}{Fan Liu}, \bibinfo{person}{Zhiyong Cheng},
  \bibinfo{person}{Changchang Sun}, \bibinfo{person}{Yinglong Wang},
  \bibinfo{person}{Liqiang Nie}, {and} \bibinfo{person}{Mohan Kankanhalli}.}
  \bibinfo{year}{2019}\natexlab{}.
\newblock \showarticletitle{User diverse preference modeling by multimodal
  attentive metric learning}. In \bibinfo{booktitle}{\emph{Proceedings of the
  27th ACM international conference on multimedia}}.
  \bibinfo{pages}{1526--1534}.
\newblock


\bibitem[Liu et~al\mbox{.}(2017)]%
        {liu2017deepstyle}
\bibfield{author}{\bibinfo{person}{Qiang Liu}, \bibinfo{person}{Shu Wu}, {and}
  \bibinfo{person}{Liang Wang}.} \bibinfo{year}{2017}\natexlab{}.
\newblock \showarticletitle{Deepstyle: Learning user preferences for visual
  recommendation}. In \bibinfo{booktitle}{\emph{Proceedings of the 40th
  international acm sigir conference on research and development in information
  retrieval}}. \bibinfo{pages}{841--844}.
\newblock


\bibitem[Luo et~al\mbox{.}(2021)]%
        {luo2021learning}
\bibfield{author}{\bibinfo{person}{Dongsheng Luo}, \bibinfo{person}{Wei Cheng},
  \bibinfo{person}{Wenchao Yu}, \bibinfo{person}{Bo Zong},
  \bibinfo{person}{Jingchao Ni}, \bibinfo{person}{Haifeng Chen}, {and}
  \bibinfo{person}{Xiang Zhang}.} \bibinfo{year}{2021}\natexlab{}.
\newblock \showarticletitle{Learning to drop: Robust graph neural network via
  topological denoising}. In \bibinfo{booktitle}{\emph{Proceedings of the 14th
  ACM International Conference on Web Search and Data Mining}}.
  \bibinfo{pages}{779--787}.
\newblock


\bibitem[Paszke et~al\mbox{.}(2019)]%
        {paszke2019pytorch}
\bibfield{author}{\bibinfo{person}{Adam Paszke}, \bibinfo{person}{Sam Gross},
  \bibinfo{person}{Francisco Massa}, \bibinfo{person}{Adam Lerer},
  \bibinfo{person}{James Bradbury}, \bibinfo{person}{Gregory Chanan},
  \bibinfo{person}{Trevor Killeen}, \bibinfo{person}{Zeming Lin},
  \bibinfo{person}{Natalia Gimelshein}, \bibinfo{person}{Luca Antiga},
  {et~al\mbox{.}}} \bibinfo{year}{2019}\natexlab{}.
\newblock \showarticletitle{Pytorch: An imperative style, high-performance deep
  learning library}.
\newblock \bibinfo{journal}{\emph{Advances in neural information processing
  systems}}  \bibinfo{volume}{32} (\bibinfo{year}{2019}).
\newblock


\bibitem[Reimers and Gurevych(2019)]%
        {reimers2019sentence}
\bibfield{author}{\bibinfo{person}{Nils Reimers} {and} \bibinfo{person}{Iryna
  Gurevych}.} \bibinfo{year}{2019}\natexlab{}.
\newblock \showarticletitle{Sentence-bert: Sentence embeddings using siamese
  bert-networks}. In \bibinfo{booktitle}{\emph{EMNLP}}.
  \bibinfo{pages}{3980--3990}.
\newblock


\bibitem[Rendle et~al\mbox{.}(2009)]%
        {rendle2009bpr}
\bibfield{author}{\bibinfo{person}{Steffen Rendle}, \bibinfo{person}{Christoph
  Freudenthaler}, \bibinfo{person}{Zeno Gantner}, {and} \bibinfo{person}{Lars
  Schmidt-Thieme}.} \bibinfo{year}{2009}\natexlab{}.
\newblock \showarticletitle{BPR: Bayesian Personalized Ranking from Implicit
  Feedback}. In \bibinfo{booktitle}{\emph{Proceedings of the Twenty-Fifth
  Conference on Uncertainty in Artificial Intelligence}}.
  \bibinfo{pages}{452--461}.
\newblock


\bibitem[Rong et~al\mbox{.}(2020)]%
        {rong2020dropedge}
\bibfield{author}{\bibinfo{person}{Yu Rong}, \bibinfo{person}{Wenbing Huang},
  \bibinfo{person}{Tingyang Xu}, {and} \bibinfo{person}{Junzhou Huang}.}
  \bibinfo{year}{2020}\natexlab{}.
\newblock \showarticletitle{Dropedge: Towards deep graph convolutional networks
  on node classification}. In \bibinfo{booktitle}{\emph{International
  Conference on Learning Representations}}.
\newblock


\bibitem[Simonyan and Zisserman(2014)]%
        {simonyan2014very}
\bibfield{author}{\bibinfo{person}{Karen Simonyan} {and}
  \bibinfo{person}{Andrew Zisserman}.} \bibinfo{year}{2014}\natexlab{}.
\newblock \showarticletitle{Very deep convolutional networks for large-scale
  image recognition}.
\newblock \bibinfo{journal}{\emph{arXiv preprint arXiv:1409.1556}}
  (\bibinfo{year}{2014}).
\newblock


\bibitem[Tao et~al\mbox{.}(2022)]%
        {tao2022self}
\bibfield{author}{\bibinfo{person}{Zhulin Tao}, \bibinfo{person}{Xiaohao Liu},
  \bibinfo{person}{Yewei Xia}, \bibinfo{person}{Xiang Wang},
  \bibinfo{person}{Lifang Yang}, \bibinfo{person}{Xianglin Huang}, {and}
  \bibinfo{person}{Tat-Seng Chua}.} \bibinfo{year}{2022}\natexlab{}.
\newblock \showarticletitle{Self-supervised Learning for Multimedia
  Recommendation}.
\newblock \bibinfo{journal}{\emph{IEEE Transactions on Multimedia}}
  (\bibinfo{year}{2022}).
\newblock


\bibitem[Wang et~al\mbox{.}(2021b)]%
        {wang2021dualgnn}
\bibfield{author}{\bibinfo{person}{Qifan Wang}, \bibinfo{person}{Yinwei Wei},
  \bibinfo{person}{Jianhua Yin}, \bibinfo{person}{Jianlong Wu},
  \bibinfo{person}{Xuemeng Song}, {and} \bibinfo{person}{Liqiang Nie}.}
  \bibinfo{year}{2021}\natexlab{b}.
\newblock \showarticletitle{DualGNN: Dual Graph Neural Network for Multimedia
  Recommendation}.
\newblock \bibinfo{journal}{\emph{IEEE Transactions on Multimedia}}
  (\bibinfo{year}{2021}).
\newblock


\bibitem[Wang et~al\mbox{.}(2021a)]%
        {wang2021denoising}
\bibfield{author}{\bibinfo{person}{Wenjie Wang}, \bibinfo{person}{Fuli Feng},
  \bibinfo{person}{Xiangnan He}, \bibinfo{person}{Liqiang Nie}, {and}
  \bibinfo{person}{Tat-Seng Chua}.} \bibinfo{year}{2021}\natexlab{a}.
\newblock \showarticletitle{Denoising implicit feedback for recommendation}. In
  \bibinfo{booktitle}{\emph{Proceedings of the 14th ACM international
  conference on web search and data mining}}. \bibinfo{pages}{373--381}.
\newblock


\bibitem[Wei et~al\mbox{.}(2020)]%
        {wei2020graph}
\bibfield{author}{\bibinfo{person}{Yinwei Wei}, \bibinfo{person}{Xiang Wang},
  \bibinfo{person}{Liqiang Nie}, \bibinfo{person}{Xiangnan He}, {and}
  \bibinfo{person}{Tat-Seng Chua}.} \bibinfo{year}{2020}\natexlab{}.
\newblock \showarticletitle{Graph-refined convolutional network for multimedia
  recommendation with implicit feedback}. In
  \bibinfo{booktitle}{\emph{Proceedings of the 28th ACM international
  conference on multimedia}}. \bibinfo{pages}{3541--3549}.
\newblock


\bibitem[Wei et~al\mbox{.}(2019)]%
        {wei2019mmgcn}
\bibfield{author}{\bibinfo{person}{Yinwei Wei}, \bibinfo{person}{Xiang Wang},
  \bibinfo{person}{Liqiang Nie}, \bibinfo{person}{Xiangnan He},
  \bibinfo{person}{Richang Hong}, {and} \bibinfo{person}{Tat-Seng Chua}.}
  \bibinfo{year}{2019}\natexlab{}.
\newblock \showarticletitle{MMGCN: Multi-modal graph convolution network for
  personalized recommendation of micro-video}. In
  \bibinfo{booktitle}{\emph{Proceedings of the 27th ACM International
  Conference on Multimedia}}. \bibinfo{pages}{1437--1445}.
\newblock


\bibitem[Wu et~al\mbox{.}(2020)]%
        {wu2020graph}
\bibfield{author}{\bibinfo{person}{Shiwen Wu}, \bibinfo{person}{Fei Sun},
  \bibinfo{person}{Wentao Zhang}, \bibinfo{person}{Xu Xie}, {and}
  \bibinfo{person}{Bin Cui}.} \bibinfo{year}{2020}\natexlab{}.
\newblock \showarticletitle{Graph neural networks in recommender systems: a
  survey}.
\newblock \bibinfo{journal}{\emph{ACM Computing Surveys (CSUR)}}
  (\bibinfo{year}{2020}).
\newblock


\bibitem[Xu et~al\mbox{.}(2022)]%
        {xu2022graph}
\bibfield{author}{\bibinfo{person}{Zhe Xu}, \bibinfo{person}{Boxin Du}, {and}
  \bibinfo{person}{Hanghang Tong}.} \bibinfo{year}{2022}\natexlab{}.
\newblock \showarticletitle{Graph sanitation with application to node
  classification}. In \bibinfo{booktitle}{\emph{Proceedings of the ACM Web
  Conference 2022}}. \bibinfo{pages}{1136--1147}.
\newblock


\bibitem[Zhang et~al\mbox{.}(2021)]%
        {zhang2021mining}
\bibfield{author}{\bibinfo{person}{Jinghao Zhang}, \bibinfo{person}{Yanqiao
  Zhu}, \bibinfo{person}{Qiang Liu}, \bibinfo{person}{Shu Wu},
  \bibinfo{person}{Shuhui Wang}, {and} \bibinfo{person}{Liang Wang}.}
  \bibinfo{year}{2021}\natexlab{}.
\newblock \showarticletitle{Mining Latent Structures for Multimedia
  Recommendation}. In \bibinfo{booktitle}{\emph{Proceedings of the 29th ACM
  International Conference on Multimedia}}. \bibinfo{pages}{3872--3880}.
\newblock


\bibitem[Zhang et~al\mbox{.}(2022)]%
        {zhang2022diffusion}
\bibfield{author}{\bibinfo{person}{Lingzi Zhang}, \bibinfo{person}{Yong Liu},
  \bibinfo{person}{Xin Zhou}, \bibinfo{person}{Chunyan Miao},
  \bibinfo{person}{Guoxin Wang}, {and} \bibinfo{person}{Haihong Tang}.}
  \bibinfo{year}{2022}\natexlab{}.
\newblock \showarticletitle{Diffusion-Based Graph Contrastive Learning for
  Recommendation with Implicit Feedback}. In
  \bibinfo{booktitle}{\emph{International Conference on Database Systems for
  Advanced Applications}}. Springer, \bibinfo{pages}{232--247}.
\newblock


\bibitem[Zhou et~al\mbox{.}(2023d)]%
        {zhou2023comprehensive}
\bibfield{author}{\bibinfo{person}{Hongyu Zhou}, \bibinfo{person}{Xin Zhou},
  \bibinfo{person}{Zhiwei Zeng}, \bibinfo{person}{Lingzi Zhang}, {and}
  \bibinfo{person}{Zhiqi Shen}.} \bibinfo{year}{2023}\natexlab{d}.
\newblock \showarticletitle{A Comprehensive Survey on Multimodal Recommender
  Systems: Taxonomy, Evaluation, and Future Directions}.
\newblock \bibinfo{journal}{\emph{arXiv preprint arXiv:2302.04473}}
  (\bibinfo{year}{2023}).
\newblock


\bibitem[Zhou et~al\mbox{.}(2023e)]%
        {zhou2023enhancing}
\bibfield{author}{\bibinfo{person}{Hongyu Zhou}, \bibinfo{person}{Xin Zhou},
  \bibinfo{person}{Lingzi Zhang}, {and} \bibinfo{person}{Zhiqi Shen}.}
  \bibinfo{year}{2023}\natexlab{e}.
\newblock \showarticletitle{Enhancing Dyadic Relations with Homogeneous Graphs
  for Multimodal Recommendation}.
\newblock \bibinfo{journal}{\emph{arXiv preprint arXiv:2301.12097}}
  (\bibinfo{year}{2023}).
\newblock


\bibitem[Zhou(2023)]%
        {zhou2023mmrec}
\bibfield{author}{\bibinfo{person}{Xin Zhou}.} \bibinfo{year}{2023}\natexlab{}.
\newblock \showarticletitle{MMRec: Simplifying Multimodal Recommendation}.
\newblock \bibinfo{journal}{\emph{arXiv preprint arXiv:2302.03497}}
  (\bibinfo{year}{2023}).
\newblock


\bibitem[Zhou et~al\mbox{.}(2023a)]%
        {zhou2023layer}
\bibfield{author}{\bibinfo{person}{Xin Zhou}, \bibinfo{person}{Donghui Lin},
  \bibinfo{person}{Yong Liu}, {and} \bibinfo{person}{Chunyan Miao}.}
  \bibinfo{year}{2023}\natexlab{a}.
\newblock \showarticletitle{Layer-refined graph convolutional networks for
  recommendation}. In \bibinfo{booktitle}{\emph{2023 IEEE 39th International
  Conference on Data Engineering (ICDE)}}. IEEE, \bibinfo{pages}{1247--1259}.
\newblock


\bibitem[Zhou et~al\mbox{.}(2022)]%
        {zhou2022bribery}
\bibfield{author}{\bibinfo{person}{Xin Zhou}, \bibinfo{person}{Shigeo
  Matsubara}, \bibinfo{person}{Yuan Liu}, {and} \bibinfo{person}{Qidong Liu}.}
  \bibinfo{year}{2022}\natexlab{}.
\newblock \showarticletitle{Bribery in Rating Systems: A Game-Theoretic
  Perspective}. In \bibinfo{booktitle}{\emph{Pacific-Asia Conference on
  Knowledge Discovery and Data Mining}}. Springer, \bibinfo{pages}{67--78}.
\newblock


\bibitem[Zhou et~al\mbox{.}(2023b)]%
        {zhou2023selfcf}
\bibfield{author}{\bibinfo{person}{Xin Zhou}, \bibinfo{person}{Aixin Sun},
  \bibinfo{person}{Yong Liu}, \bibinfo{person}{Jie Zhang}, {and}
  \bibinfo{person}{Chunyan Miao}.} \bibinfo{year}{2023}\natexlab{b}.
\newblock \showarticletitle{Selfcf: A simple framework for self-supervised
  collaborative filtering}.
\newblock \bibinfo{journal}{\emph{ACM Transactions on Recommender Systems}}
  \bibinfo{volume}{1}, \bibinfo{number}{2} (\bibinfo{year}{2023}),
  \bibinfo{pages}{1--25}.
\newblock


\bibitem[Zhou et~al\mbox{.}(2023c)]%
        {zhou2023bootstrap}
\bibfield{author}{\bibinfo{person}{Xin Zhou}, \bibinfo{person}{Hongyu Zhou},
  \bibinfo{person}{Yong Liu}, \bibinfo{person}{Zhiwei Zeng},
  \bibinfo{person}{Chunyan Miao}, \bibinfo{person}{Pengwei Wang},
  \bibinfo{person}{Yuan You}, {and} \bibinfo{person}{Feijun Jiang}.}
  \bibinfo{year}{2023}\natexlab{c}.
\newblock \showarticletitle{Bootstrap latent representations for multi-modal
  recommendation}. In \bibinfo{booktitle}{\emph{Proceedings of the ACM Web
  Conference 2023}}. \bibinfo{pages}{845--854}.
\newblock


\end{thebibliography}

\appendix
\section{Proof of LEMMA 4.1}
\label{append:lemma}
Let $s'_{ij}$ represent the element in the adjacency matrix of LATTICE and let $s_{ij}$ represent the corresponding element in FREEDOM. We can derive that $s_{ij} = 1 \geq s'_{ij}$. As defined in the paper, $\bm{D}'=\mathrm{diag}(d'_{11}, \cdots, d'_{nn})$ is a diagonal matrix where each entry on the diagonal is equal to the row-sum of the adjacency matrix $d'_{ii} = \sum_j s'_{ij}$. The normalized adjacency matrix $\widetilde{\bm{S}'}$ in LATTICE can be denoted as follows:

\begin{align*}
	\widetilde{\bm{S}'} = 
	\begin{bmatrix}
		1/d'_{11}       & {s'}_{12}/\sqrt{d'_{11}d'_{22}}  & \dots & {s'}_{1n}/\sqrt{d'_{11}d'_{nn}} \\
		{s'}_{21}/\sqrt{d'_{22}d'_{11}} & 1/d'_{22} & \dots & {s'}_{2n}/\sqrt{d'_{22}d'_{nn}} \\
		\vdots &  \vdots &  & \vdots \\
		{s'}_{n1}/\sqrt{d'_{nn}d'_{11}} &  {s'}_{n2}/\sqrt{d'_{nn}d'_{22}} & \dots & 1/d'_{nn}
	\end{bmatrix}
\end{align*}

Analogously, we can derive the normalized adjacency matrix in \fm{}.
We denote the elements in the normalized adjacency matrices of $\widetilde{\bm{S'}}$ and $\widetilde{\bm{S}}$ as $\tilde{s'}_{ij}$ and $\tilde{s}_{ij}$, respectively.
As stated in the paper, the function $g(\tilde{s'}_{ij}) = \tilde{s'}_{ij}/(\tilde{s'}_{ij} +\sum_{k\neq i} \tilde{s'}_{kj})$ is monotonically increasing. To maximize Eq.~(13) in the paper, $\tilde{s'}_{ij}$ must also be maximized. From the above equations, we can deduce that the maximum value for either $\widetilde{\bm{S}}$ or $\widetilde{\bm{S'}}$ lies on its diagonal. Therefore, we can derive that $max_{i,j}\tilde{s}_{ij} = 1/d_{ii} = 1/\sum_j s_{ij} \leq max_{i,j}\tilde{s'}_{ij} = 1/d'_{ii} = 1/\sum_j s'_{ij}$.

\end{document}